\documentclass[11pt]{article}
\usepackage[margin=1in]{geometry}
\usepackage{enumerate}
\usepackage[vlined,noend,boxed,algoruled]{algorithm2e}

\usepackage{amsmath,amsthm,amssymb}
\usepackage{graphicx}
\usepackage{paralist}
\usepackage{pgf}
\usepackage{tikz}
\usetikzlibrary{arrows,automata}

\newtheorem{theorem}{Theorem}[section]

\newtheorem{definition}{Definition}

\newtheorem{lemma}[theorem]{Lemma}
\newtheorem{corollary}[theorem]{Corollary}

\newcommand{\N}{\mathbb{N}}

\newcommand{\E}{\mathbb{E}}
\newcommand{\Z}{\mathbb{Z}}
\newcommand{\BO}{\mathcal{O}}

\begin{document}

\title{Trade-offs between Selection Complexity and Performance when Searching the Plane without Communication}
\date{}
\author{
  Christoph Lenzen\\
  \texttt{clenzen@csail.mit.edu}
  \and
  Nancy Lynch\\
  \texttt{lynch@csail.mit.edu}
  \and
  Calvin Newport\\
  \texttt{cnewport@cs.georgetown.edu}
  \and
  Tsvetomira Radeva\\
  \texttt{radeva@csail.mit.edu}
}
\maketitle

\begin{abstract}
We consider the ANTS problem [Feinerman et al.] in which a group of agents collaboratively search for a target in a two-dimensional plane. Because this problem is inspired by the behavior of biological species, we argue that in addition to studying the {\em time complexity} of solutions it is also important to study the {\em selection complexity}, a measure of how likely a given algorithmic strategy is to arise in nature due to selective pressures. In more detail, we propose a new selection complexity metric $\chi$, defined for algorithm ${\cal A}$ such that $\chi({\cal A}) = b + \log \ell$, where $b$ is the number of memory bits used by each agent and $\ell$ bounds the fineness of available probabilities (agents use probabilities of at least $1/2^\ell$). In this paper, we study the trade-off between the standard performance metric of speed-up, which measures how the expected time to find the target improves with $n$, and our new selection metric. 

In particular, consider $n$ agents searching for a treasure located at (unknown) distance $D$ from the origin (where $n$ is sub-exponential in $D$). For this problem, we identify $\log \log D$ as a crucial threshold for our selection complexity metric. We first prove a new upper bound that achieves a near-optimal speed-up of $(D^2/n +D) \cdot 2^{\BO(\ell)}$ for $\chi({\cal A}) \leq 3 \log \log D + \BO(1)$. In particular, for $\ell \in \BO(1)$, the speed-up is asymptotically optimal. By comparison, the existing results for this problem [Feinerman et al.] that achieve similar speed-up require $\chi({\cal A}) = \Omega(\log D)$. We then show that this threshold is tight by describing a lower bound showing that if $\chi({\cal A}) < \log \log D - \omega(1)$, then with high probability the target is not found within $D^{2-o(1)}$ moves per agent. Hence, there is a sizable gap to the straightforward $\Omega(D^2/n + D)$ lower bound in this setting.

\end{abstract}

\thispagestyle{empty}
\setcounter{page}{0}

\pagebreak

\section{Introduction}
\label{sec:intro}

It is increasingly accepted by some biologists and computer scientists that the tools of distributed computation can improve our understanding of distributed biological processes~\cite{feinerman12disc, feinerman13, feinerman12podc}. A standard approach is to translate a biological process of interest (e.g., ant foraging~\cite{feinerman12disc, feinerman12podc} or sensory organ pre-cursor selection~\cite{afek11}) into a formal problem in a distributed computing model, and then prove upper and lower bounds on the problem. The aim is to use these bounds to gain insight into the behavior of the motivating biological process.

A recognized pitfall of this approach is \emph{incongruous analysis}, in which the theoretician focuses on metrics relevant to computation but not biology, or ignores metrics relevant to biology but not to computation.  Motivated by this pitfall, this paper promotes the use of {\em selection complexity} metrics for studying biologically-inspired distributed problems. Unlike standard metrics from computation, which tend to focus only on performance, selection complexity metrics instead attempt to measure the difficulty of a given algorithmic strategy arising in nature as the result of selective pressures. Roughly speaking, a solution with low selection complexity should be more likely to arise in nature than a solution with high selection complexity.

We argue that theoreticians studying biologically-inspired problems should evaluate solutions in terms of selection complexity in addition to focusing on standard performance metrics; perhaps even measuring the trade-off between the two classes of metrics. This paper provides a case study of this approach by fixing a standard biology-inspired problem and new selection complexity metric, and then bounding the trade-off between performance and selection complexity with respect to this metric. In doing so, we also obtain results regarding concurrent non-uniform random walks that are of independent mathematical interest.

We recognize that most papers on biology-inspired distributed problems implicitly address selection complexity in their fixed model constraints. Restricting agents to not have access to communication in the search problem, for example, is a constraint that likely lowers the selection complexity of solutions in the model. What is new about our approach is that we are capturing such complexity in a variable metric, allowing us to study the trade-offs between algorithmic power and performance more generally. This can provide insights beyond those gained by characterizing the capabilities of a given static set of constraints.

In this paper, we focus on the problem of $n$ probabilistic non-communicating agents collaboratively searching for a target in a two-dimensional grid placed at (unknown) distance $D$ (measured in number of hops in the grid) from the origin. We assume that $n$ is sub-exponential in $D$.\footnote{Note that an exponential number of agents finds the target quickly even if they employ simple random walks.} This problem is described and analyzed in recent work by Feinerman et al.~\cite{feinerman12podc} (referred to as the ANTS problem). The authors in ~\cite{feinerman12podc} argue that it provides a good approximation of insect foraging, and represents a useful intersection between biological behavior and distributed computation. The analysis in~\cite{feinerman12podc} focuses on the {\em speed-up} performance metric, which measures how the expected time to find the target improves with $n$. The authors describe and analyze search algorithms that closely approximate the straightforward $\Omega(D + D^2/n)$ lower bound for finding a target placed at distance $D$ from the origin.
 
 \paragraph{Selection metric motivation.} We consider the selection complexity metric $\chi$, which captures the bits of memory and probabilistic range used by a given algorithm. This combined metric is motivated by the fact that memory can be used to simulate small probability values, and such values give more power to algorithms, e.g.\ permitting longer directed walks with a given amount of memory. In more detail, for algorithm ${\cal A}$, we define $\chi({\cal A}) = b + \log{\ell}$, where $b$ is the number of bits of memory required by the algorithm (note, $b = \log{|S|}$, where $S$ is the state set of the state machine representation of ${\cal A}$), and $\ell$ is the smallest value such that all probabilities used in ${\cal A}$ are bounded from below by $1/2^\ell$. In Section \ref{sec:upper} and Section \ref{sec:lower}, we show that the choice of the selection metric arises naturally from the analysis of our algorithms and the lower bound.

We conjecture that, from a biological point of view, it is reasonable to assume that large values of $\ell$ are associated with higher selection complexity. Clearly, algorithms relying on small probabilities are more sensitive to additive disturbances of the probability values. Hence, creating a small probability based on a single event is harder to accomplish, since the  event must not only have a strong bias towards one outcome, but also be well protected against influencing factors (like temperature, noise, etc.). On the other hand, using multiple independent events to simulate one with larger bias (also known as probability boosting) constitutes a hidden cost. Our model and algorithms make this cost explicit, by accounting for it in terms of the memory needed for counting such events.

\paragraph{Results.}
In this paper, we generalize the problem of~\cite{feinerman12podc} by now also considering the selection complexity metric $\chi$. 
We identify $\log\log{D}$, for target distance $D$, as a crucial threshold for the $\chi$ metric when studying the achievable speed-up in the foraging problem. In more detail, our lower bound proves that for any algorithm ${\cal A}$ such that $\chi({\cal A}) \leq \log\log{D} - \omega(1)$, there is a placement of the treasure at distance $D$ such that the probability that ${\cal A}$ finds the treasure in less than $D^{2-o(1)}$ moves per agent is polynomially small in $D$, and the probability of finding a target placed randomly within this distance is $o(1)$. The {\em speed-up} in this case is bounded from above by $\min\{n, D^{o(1)}\}$, as opposed to the optimal speed-up of $\min\{n,D\}$. At the core of our lower bound is a novel analysis of recurrence behavior of small Markov chains with probabilities of at least $1/2^\ell$.

Concerning upper bounds, we note that the foraging algorithms in~\cite{feinerman12podc} achieve near-optimal speed-up in $n$, but their selection complexity, as measured by $\chi({\cal A})$, is higher than the $\log\log{D}$ threshold identified by our lower bound: these algorithms require sufficiently fine-grained probabilities and enough memory to randomly generate and store, respectively, coordinates up to distance at least $D$ from the origin; this entails $\chi({\cal A})\geq \log{D}$. In this paper, we seek upper bounds that work for $\chi({\cal A}) \approx \log\log{D}$, the minimum value for which good speed-up is possible. With this in mind, we begin by describing and analyzing a very simple algorithm that is non-uniform in $D$ (agents know the value of $D$) and has asymptotically optimal expected running time. It illustrates our main ideas of walking up to certain points in the plane while counting approximately, thus using little memory, and showing that this is sufficient for searching the plane efficiently. This algorithm uses a value of $\chi = \log \log D + \BO(1)$, which matches our lower bound result for $\chi$ up to factor $1 + o(1)$.

We generalize the ideas used in our simple algorithm to derive a solution that is uniform in $D$. The main idea is to start with some estimate of $D$ and keep increasing it while executing a corresponding version of our simple search algorithm described above for each such estimate. Our uniform algorithm solves the problem in  $\BO (D^2/n + D) \cdot 2^{\BO(\ell)}$ moves per agent in expectation (if $\ell = \BO(1)$, the algorithm matches the $\Omega(D^2/n + D)$ lower bound), for $\chi({\cal A}) \leq 3 \log\log{D} + \BO(1)$. We remark that the increased running time is due to the fact that in order to keep the value of $\chi$ small, we increase our estimate of $D$ by a factor of $2^{\BO(\ell)}$ in each step, which may result in ``overshooting'' the correct distance by factor $2^{\BO(\ell)}$. Note that this suboptimal expected running time arises from enforcing $o(\log \log D)$ memory bits; otherwise, one is always free to use the only constant probabilities.

  \paragraph{Discussion.}
  An interesting question that arises from our results is the trade-off between $b$ and $\ell$ in the definition of $\chi({\cal A})$: roughly speaking, more bits of memory might be of greater utility than having access to smaller probabilities. This seems intuitive given that smaller probability values can be simulated using additional memory  (e.g., to simulate a coin that returns heads with probability $1/2^k$, flip a uniform coin $k$ times while storing the number of coin tosses in the additional memory), but in general more precise probabilities cannot be used to simulate additional memory.
  
From a biological perspective, we do not claim that $\chi$ is necessarily the {\em right} selection metric to use in studying such problems. We chose it because $b$ and $\ell$ seem to be important factors in search, and  they are potentially difficult to increase in nature. However, we recognize that the refinement and validation of such metrics require close collaboration with evolutionary biologists. In this paper, our main goal is to advertise the selection complexity approach as a promising tool for studying biology-inspired problems.

From a mathematical perspective, we emphasize that our lower bound result, in particular, is of independent interest. It is known that uniform random walks do not provide substantial speed-up in the plane searching problem~\cite{alon08}; the speed-up is bounded by $\min\{\log{n},D\}$. Our lower bound generalizes this observation from uniform random walks to probabilistic processes with bounded probabilities and small state complexities. 
 
\paragraph{Related Work.}

This work was initially inspired by the results in \cite{feinerman12disc} and \cite{feinerman12podc}, which originally introduced the problem studied here. More precisely, in \cite{feinerman12podc} the authors present an algorithm to find the target in optimal expected time $\BO(D^2/n + D)$, assuming that each agent in the algorithm knows the number $n$ of agents (but not $D$). For unknown $n$, they show that for every constant $\epsilon > 0$, there exists a uniform search algorithm that is $\BO(\log^{1+\epsilon} k)$-competitive, but there is no uniform search algorithm that is $\BO(\log k)$-competitive. In \cite{feinerman12disc}, Feinerman et al. provide multiple lower bounds on the advice size (number of bits of information the ants are given prior to the search), which can be used to store the value $n$, some approximation of it, or any other information. In particular, they show that in order for an algorithm to be $\BO(\log^{1-\epsilon})$-competitive, the ants need advice size of $\Omega(\log \log n)$ bits. Note that this result also implies a lower bound of $\Omega(\log \log n)$ bits on the total size of the memory of the ants, but only under the condition that optimal speed-up is required. Our lower bound is stronger in that we show that there is an exponential gap of $D^{1-o(1)}$ for the maximum speed-up (with a sub-exponential number of agents in $D$). Similarly, the algorithms in \cite{feinerman12podc} need at least $\BO(\log D)$ bits of memory, as contrasted with our algorithm that uses $b \leq 3 \log \log D + \BO(1)$ bits of memory. 

 Searching and exploration of various types of graphs by single and multiple agents are widely studied in the literature. Several works study the case of a single agent exploring directed graphs \cite{albers00, bender98, deng90}, undirected graphs \cite{panaite19, reingold05}, or trees \cite{diks02, gasieniec07}. Out of these, the following papers have restrictions on the memory used in the search: \cite{gasieniec07} uses $\BO(\log n)$ bits to explore an $n$-node tree, \cite{bender98} studies the power of a pebble placed on a vertex so that the vertex can later be identified, \cite{diks02} shows that $\Omega(\log \log n)$ bits of memory are needed to explore some $n$-node trees, and \cite{reingold05} presents a $\log$-space algorithm for $st$-connectivity. There have been works on graph exploration with multiple agents \cite{alon08, emek13, fraigniaud06}; while \cite{alon08} and \cite{fraigniaud06} do not include any memory bounds, \cite{emek13} presents an optimal algorithm for searching in a grid with constant memory and constant-sized messages in a model, introduced in \cite{emek13podc}, of very limited computation and communication capabilities. It should be noted that even though these models restrict the agents' memory to very few bits, the fact that the models allow communication makes it possible to simulate larger memory. 
 
 So far, in the above papers, we have seen that the metrics typically considered by computer scientists in graph search algorithms are mostly the amount of memory used and the running time. In contrast, biologists look at a much wider range of models and metrics, more closely related to the physical capabilities of the agents. For example, in \cite{arbilly10} the focus is on the capabilities of foragers to learn about different new environments, \cite{giraldeau00} considers the physical fitness of agents and the abundance and quality of the food sources, \cite{harkness85} considers interesting navigational capabilities of ants and assumes no communication between them, \cite{holder87} measures the efficiency of foraging in terms of the energy over time spent per agent, and \cite{robinson05} explores the use of different chemicals used by ants to communicate with one another.

\paragraph{Organization.}

In Section \ref{sec:model}, we present our system model assumptions and formally define the search problem and both the performance and selection metrics that we use to evaluate our algorithms. In Section \ref{sec:upper}, we present our algorithms, starting with a very simple non-uniform algorithm in Section \ref{sec:non-uniform} illustrating our main approach. In Section \ref{sec:uniform}, we generalize this approach to algorithms that are uniform in $D$. In Section \ref{sec:lower}, we present a lower bound that matches our upper bounds in terms of the selection metric $\chi$. We conclude by discussing some assumptions and possible extensions of our work in Section \ref{sec:discussion}. The appendix contains some definitions and math preliminaries used throughout the technical sections of the paper.

\section{Model}
\label{sec:model}

Our model is similar to the models considered in \cite{feinerman12disc, feinerman12podc}. We consider an infinite two-dimensional square grid with coordinates in $\Z^2$. The grid is to be explored by $n\in \N$ identical, non-communicating, probabilistic agents. Each agent is always located at a point on the grid. Agents can move in one of four directions, to one of the four adjacent grid points, but they have no information about their current location in the grid. Initially all agents are positioned at the origin. We also assume that an agent can return to the origin, and for the purposes of this paper, we assume this action is based on information provided by an oracle. In this case, the agent returns on a shortest path in the grid that keeps closest to the straight line connecting the origin to its current position. Note that the return path is at most as long as the path of the agent away from the origin; therefore, since we are interested in asymptotic complexity, we ignore the lengths of the return paths in our analysis. Next, we give a formal description of our model.

\paragraph{Agents.}
Each agent is modeled as a probabilistic finite state automaton; since agents are identical, so are their state automata. Each automaton is a tuple $(S, s_0, \delta)$, where $S$ is a set of states, state $s_0 \in S$ is the unique starting state, and $\delta$ is a transition function  $\delta: S \to \Pi$, where $\Pi$ is a set of discrete probability distributions. Thus, $\delta$ maps each state $s \in S$ to a discrete probability distribution $\delta(s) = \pi_s$ on $S$, which denotes the probability of moving from state $s$ to any other state in $S$. 

For our lower bound in Section \ref{sec:lower}, it is convenient to use a Markov chain representation of each agent. Therefore, we can express each agent as a Markov chain with transition matrix $P$, such that for each $s_1, s_2 \in S$, $P[s_1][s_2] = \pi_{s_1}(s_2)$, and start state $s_0 \in S$. 

In addition to the Markov chain that describes the evolution of an agent's state, we also need to characterize its movement on the grid. We define a  labeling function $M : S \to \{\text{up, down, right, left,} \\ \text{origin, none}\}$ mapping each state $s\in S$ to an action the agent performs on the grid. For simplicity, we require $M(s_0) =\,$origin. Using this labeling function, any sequence of states $(s_i\in S)_{i\in \N}$ is mapped to a sequence of moves in the grid $(M(s_i))_{i\in \N}$ where $M(s_i)=\,$none denotes no move in the grid (i.e., $s_i$ does not contribute to the derived sequence of moves) and $M(s_i)=\,$origin means that the agent returns to the origin, as described above. 

\paragraph{Executions.}
An execution of an algorithm for some agent is given by a sequence of states from $S$, starting with state $s_0$, and coordinates of the associated movements on the grid derived from these states. Formally, an execution is defined as $(s_0, (x_0, y_0), s_1, (x_1, y_1), s_2, (x_2, y_2), \cdots)$, where $s_0 \in S$ is the start state, $(x_0, y_0) = (0,0)$, and for each $i \geq 1$, applying the move $M(s_{i+1})$ to point $(x_i, y_i)$ results in point $(x_{i+1}, y_{i+1})$. For example, if $M(s_{i+1}) = \text{up}$, then $x_{i+1}=x_i$ and $y_{i+1} = y_i + 1$. For the case where $M(s_{i+1}) = \text{none}$, we define $x_i = x_{i+1}$ and $y_i = y_{i+1}$, and for $M(s_{i+1}) = \text{origin}$, we define $(x_{i+1}, y_{i+1}) = (0,0)$. In other words, we ignore the movement of the agent on the way back to the origin, as mentioned earlier in this section. 

An execution of an algorithm with $n$ agents is just an $n$-tuple of executions of single agents. For our analysis of the lower bound, it is useful to assume a synchronous model. So, we define a \emph{round} of an execution to consist of one transition of each agent in its Markov chain. Note that we do not use such synchrony for our algorithms.

So far, we have described a linear execution of an algorithm with $n$ agents. In order to consider probabilistic executions, note that the Markov chain $(S,P)$ induces a probability distribution of executions in a natural way, by performing an independent random walk on $S$ with transition probabilities given by $P$ for each of the $n$ agents.

\paragraph{Problem Statement.}
The goal is to find a target located at some vertex at distance (measured in terms of the max-norm) at most $D$ from the origin in as few expected moves as possible. Note that measuring paths in terms of the max-norm gives is a constant-factor approximation of the actual hop distance. We will consider both uniform and non-uniform algorithms with respect to $D$; that is, the agents may or may not know the value of $D$. 

It is easy to see (also shown in \cite{feinerman12podc}) that the expected running time is $\Omega(D + D^2/n)$ even if agents know $n$ and $D$ and they can communicate with each other. This bound can be matched if the agents know a constant-factor approximation of $n$ \cite{feinerman12podc}, but as mentioned in Section \ref{sec:intro}, the value of the selection metric $\chi$ (introduced below) in that specific algorithm is $\Omega(\log D)$. For simplicity, throughout this paper we will consider algorithms that are non-uniform in $n$, i.e., the agents' state machine depends on $n$. We can apply a technique from \cite{feinerman12podc}, that the authors use to make their algorithms uniform in $n$, in order to generalize our results and obtain an algorithm that is uniform in both $D$ and $n$.

\paragraph{Metrics. } For the problem defined above, we consider both a performance and a selection metric and study the trade-off between the two. 
We will use the term \emph{step} of an agent interchangeably with a transition of the agent in the Markov chain. We define a \emph{move} of the agent to be a step that the agent performs in its Markov chain resulting in a state labeled up, down, left, or right. 

For our performance metric, we focus on the asymptotic running time in terms of $D$ and $n$; more precisely, we are interested in the expected value of the metric $M_{\text{moves}}$: the minimum over all agents of the number of moves of the agent until it finds the target. Note that for the performance metric we exclude states labeled \emph{none} and \emph{origin} in an execution of an agent; we consider the \emph{none} states to be part of an agent's local computation, and we already argued that the \emph{origin} states increase the running time by at most a factor of two. For our lower bound, it is useful to define a similar metric in terms of the steps of an agent. We define the metric $M_{\text{steps}}$ to be the minimum over all agents of the number of steps of the agent until it finds the target.

The selection metric of a state automaton (and thus a corresponding algorithm) is defined as $\chi({\cal A}) = b + \log \ell$, where $b:=\lceil\log |S|\rceil$ is the number of bits required to encode all states from $S$ and $1/2^{\ell}$ is a lower bound on $\min\{P[s,s']\,|\,s,s'\in S\wedge P[s,s']\neq 0\}$, the smallest non-zero probability value used by the algorithm. We further motivate this choice in Section \ref{sec:upper}, where we describe different trade-offs between the performance metric and the values of $b$ and $\ell$.

\section{Algorithms}
\label{sec:upper}
In this section, we begin by describing a non-uniform algorithm in $D$ that finds the target in asymptotically optimal time. The main purpose for presenting this algorithm is to illustrate our main techniques in a very simple setting. This algorithm uses probability values of the form $1/D$, which can easily be simulated using only biased coins that show heads with probability $1/2^\ell$ for any $\ell$ such that $\log D$ is an integer multiple of $\ell$. We show that the target can be found in asymptotically optimal time using $b = \log \log D - \log \ell + 3$ bits of memory. 

We then generalize this algorithm to work for the case of unknown $D$. This ensures that closer targets are found faster by the algorithm than targets that are far away. The way we achieve this is by starting with an estimate of $D$ equal to $2$ and repeatedly increasing it until the target is found. For each such estimate we execute the non-uniform algorithm. Since $D$ is not known by the algorithm anymore, we cannot easily pick fixed values for some of the parameters we use in the algorithm in order to guarantee asymptotically optimal results for all possible values of $D$. Therefore, in our general algorithm the expected number of moves for the first agent to find the target becomes $(D^2/n + D) \cdot 2^{\BO(\ell)}$ for $\chi = 3 \log \log D + \BO(1)$. Hence, for $\ell=\BO(1)$ the algorithm is asymptotically optimal with respect to both metrics, and we achieve non-trivial speed-up of $\min\{n,D\}/D^{o(1)}$ for any $\ell\in o(\log D)$ (i.e., $\omega(1)$ bits of memory).

\subsection{Non-uniform Algorithm}
\label{sec:non-uniform}

In this section we present an algorithm in which the value of $D$ is available to the algorithm. We assume that $D > 1$; the cases of $D=0$ and $D=1$ are straightforward. Our general approach is the following: each agent chooses a vertical direction (up or down) with probability $1/2$, walks in that direction for a random number of steps that depends on $D$, then does the same for the horizontal direction, and finally returns to the origin and repeats this process. We show that the minimum over all agents of the expected number of moves of the agent to find a target at distance up to $D$ from the origin is at most $\BO(D^2/n + D)$. 

Let coin $C_p$ denote a coin that shows tails with probability $p$. Using this convention, the pseudocode of this simple routine is given in Algorithm~\ref{algo:non-uni-loop}, accompanied by a state machine representation showing that the algorithm can be implemented using only three bits of memory. Later in this section we show that a slightly modified version of the algorithm guarantees that $\chi = \log \log D +3$.

\begin{figure}[t]
\begin{minipage}{\textwidth}
\begin{minipage}{.44\textwidth}
\begin{algorithm}[H]
\caption{Non-uniform search.}
\label{algo:non-uni-loop}
\While{\text{true}}{
	\If{coin $C_{1/2}$ shows heads}{
  		\While{coin $C_{1/D}$ shows heads}{
  			move up\phantom{Xp}
		}
	}
	\Else{
  		\While{coin $C_{1/D}$ shows heads}{
  			move down\phantom{Xp}
		}
	}
	\If{coin $C_{1/2}$ shows heads}{
  		\While{coin $C_{1/D}$ shows heads}{
  			move left\phantom{Xp}
		}
	}
	\Else{
  		\While{coin $C_{1/D}$ shows heads}{
  			move right\phantom{Xp}
		}
	}
	return to the origin
  }
\end{algorithm}
\end{minipage}
\hfill
\begin{minipage}{.40\textwidth}
\begin{center}
\scalebox{.8}{
\begin{tikzpicture}[->,>=stealth',shorten >=1pt,auto, node distance=3.6cm, semithick]
  \tikzstyle{every state}=[fill=none,draw=black,text=black, minimum size=1.2cm]

  \node[state] (A) {origin};
  \node[state] (B) [above of=A] {up};
  \node[state] (D) [below of=A] {down};
  \node[state] (C) [right of=A] {right};
  \node[state] (E) [left of=A] {left};

  \path (A) edge [loop above] node {$\frac{1}{D^2}$} (A)
            edge [bend left] node {$\frac{1}{2} \left(1 - \frac{1}{D}\right)$} (B)
            edge [bend left] node {$\frac{1}{2} \left(1 - \frac{1}{D}\right)$} (D)
            edge [bend left] node {$\frac{1}{2D}\left(1 - \frac{1}{D}\right)$} (C)
			edge [bend left] node {$\frac{1}{2D}\left(1 - \frac{1}{D}\right)$} (E)            
            
        (B) edge [loop above] node [right,shift={(.2cm,-.3cm)}] {$1 - \frac{1}{D}$} (B)
        	edge [bend left] node {$\frac{1}{D}$} (A)
        	edge [bend left] node [above,xshift=.3cm] {$\frac{1}{2D}$} (C)
        	edge [bend right] node [above,xshift=-.3cm] {$\frac{1}{2D}$} (E)
           
        (C) edge [loop right] node [above,yshift=.3cm] {$1 - \frac{1}{D}$} (A)
            edge [bend left]  node {$\frac{1}{D}$} (A)
        
        (D) edge [loop below] node [right,shift={(.2cm,.3cm)}] {$1 - \frac{1}{D}$} (D)
        	edge [bend left] node {$\frac{1}{D}$} (A)
        	edge [bend right] node [below,xshift=.3cm] {$\frac{1}{2D}$} (C)
        	edge [bend left] node [below,xshift=-.3cm] {$\frac{1}{2D}$} (E)
            
        (E) edge [loop left] node [above,yshift=.3cm] {$1 - \frac{1}{D}$} (A)
            edge [bend left]  node {$\frac{1}{D}$} (A);
\end{tikzpicture}
}

State machine representation of Algorithm~\ref{algo:non-uni-loop}. State names match the values of the labeling function.

\end{center}
\end{minipage}
\end{minipage}
\end{figure}

Denote by $R$ the expected number of moves for an agent to complete an iteration of the outer loop of the algorithm. By independence of the random choices, this expectation does not depend on the considered iteration. Also, since agents are identical and independent, the value of $R$ does not depend on the choice of agent either.
\begin{lemma}
\label{lem:exp_D}
$R \leq 2D$.
\end{lemma}
\begin{proof}
In each iteration, an agent performs one move up or down for each consecutive toss of coin $C_{1/D}$ showing heads, and then one move right or left for each consecutive toss of coin $C_{1/D}$ showing heads. Each of these walks is $D$ steps long in expectation, so it follows that $R \leq 2D$.
\end{proof}

One may think that all we need to do now is to compute the expected number of iterations until some agent finds the target and multiply that by the expected number of moves to complete an iteration. However, this does not quite work, because the time to complete an iteration is not independent of whether the target is found or not. For instance, if the target is located at $(1,0)$, then an agent finding the target cannot move up or down before going right, a constraint that decreases the expected number of steps taken in the iteration. 

However, since in each iteration an agent is, in fact, quite unlikely to find the target, the expectation cannot be affected a lot. For simplicity, we will use a fairly loose bound; we are interested in the asymptotic time complexity only, which is not affected. To this end, denote by $\hat{R}$ the expected number of moves on the grid an agent makes during an iteration conditioning on the event that the agent does not find the target in this iteration. Note that this event is only defined if the target is not located at the origin $(0,0)$ and is automatically found right away; without loss of generality, we will assume that this is not the case.

\begin{lemma}
\label{lem:condition}
	 $\hat{R} \leq 2 R$.
\end{lemma}
\begin{proof}

First, we bound the probability of an agent not finding the target in a given iteration of the main loop from below. Suppose that the target is located at $(x,y)\in \Z^2\setminus \{(0,0)\}$. If $y>0$, the pseudocode shows that with probability $1/2$ coin $C_{1/D}$ shows tails, so the agent does not move up, and consequently, with probability at least $1/2$ it does not find the target in this iteration. Symmetrically, the agent does not move down with probability $1/2$ and misses the target if $y<0$, and it does not move left or right with probability $1/2$ each and misses the target if $x\neq 0$. Overall, the target is missed in a given iteration with probability at least $1/2$. We partition the probability space into the events: (1) the target is found during the given iteration, and (2) the target is not found during the given iteration. From the law of total expectation applied to this partition, it follows that 
\begin{eqnarray*}
R & = & P[\text{target is found in the given iteration}] \\
		& \cdot & \E[\text{moves in the given iteration } | \text{ target is found in the given iteration}] \\
		& + & P[\text{target is \emph{not} found in the given iteration}] \\
		& \cdot & \E[\text{moves in the given iteration } | \text{ target is \emph{not} found in the given iteration}] \\
		& \leq & \hat{R} \cdot 1/2
\end{eqnarray*}

In the last step, we bound the value of $R$ by ignoring the first term in the sum, and using the fact above that the probability to miss the target in any given iteration is at least $1/2$. We conclude that $\hat{R} \leq 2 R$.
\end{proof}

\begin{lemma}
\label{lem:ria}
Let $R_{i,a}$ be the expected number of moves until a fixed agent $a$ finds the target in iteration $i$, conditioning on the fact that agent $a$ finds the target in iteration $i$ and not in any previous iteration. Then, $R_{i,a} \leq 4iD$.
\end{lemma}
\begin{proof}
 Since we are conditioning on the fact that agent $a$ finds the target in iteration $i$ and not in any previous iteration, we know that agent $a$ completes the first $i-1$ iterations of the main loop and then moves to the target. The agent requires at most $(i-1)\cdot \hat{R}$ moves to complete the first $i-1$ iterations of the loop in expectation because we know the target is not found by  agent $a$ in any of these iterations. Afterwards, in iteration $i$, it will move to the target, which takes at most $2D$ moves. Thus, by Lemmas~\ref{lem:exp_D} and~\ref{lem:condition}, 

\begin{equation*}
R_{i,a} \leq (i-1) \hat{R} + 2D \leq 2 R (i-1) + 2D \leq 4D(i-1) + 2D \leq 4iD.
\end{equation*}
\end{proof}

Having examined the expected number of moves for a single agent to complete an iteration, next we calculate how likely it is that all agents miss the target in a single iteration of the main loop. In the following lemmas and theorems we  switch from considering a probability distribution for one agent to considering a probability distribution for all $n$ agents.

\begin{lemma}
\label{lem:iterations}
	Denote by $q$ the probability that no agent finds the target when each agent executes one iteration of the main loop. It holds that $q \leq \max\{1-\Omega(n/D),1/2\}$, where $n$ is the total number of agents.
\end{lemma}

\begin{proof}
	Suppose the target is positioned at grid point $(x,y)\in \{0,\ldots,D\}^2$ and consider a single agent performing an iteration of its main loop. With probability $1/4$, it moves up and right. The walk up will halt after exactly $x$ steps with probability
\begin{equation*}
\left(1-\frac{1}{D}\right)^x\frac{1}{D}\geq \left(1-\frac{1}{D}\right)^{D}\frac{1}{D} \geq \frac{1}{4D}.
\end{equation*}
The walk right will perform $D \geq y$ steps with probability at least $(1-1/D)^D \geq 1/4$. Hence, in each iteration, each agent finds the target with probability at least $1/(64D)$. Analogously, the same holds for a target located at $(-x,y)$, $(-x,-y)$, and $(x,-y)$. 

Since iterations performed by different agents are independent of each other, it follows that $q \leq (1-1/(64D))^n$. We use the binomial expansion of the right hand side for the case when $n > 64D$:

\begin{eqnarray*}
	\left(1-\frac{1}{64D}\right)^n &=& 1 - \frac{n}{64D} + \frac{(n-1)n}{2!} \frac{1}{(64D)^2} - \frac{(n-2)(n-1)n}{3!} \frac{1}{(64D)^3} + \cdots \\
	& \leq & 1 - \frac{n}{64D} + \frac{(n-1)n}{2!} \frac{1}{(64D)^2} + \frac{(n-3)(n-2)(n-1)n}{4!} \frac{1}{(64D)^4} + \cdots \\
	& \leq & 1 - \frac{n}{64D} + \frac{n^2}{2!} \frac{1}{(64D)^2} + \frac{n^4}{4!} \frac{1}{(64D)^4} + \cdots \\
	& \leq & 1 - \frac{n}{64D} + \frac{(64D)^2}{2!(64D)^2} + \frac{(64D)^4}{4!(64D)^4} + \cdots
	= 1 - \frac{n}{64D} + \sum_{i=2}^{\infty} \frac{1}{i!} = 1 - \Omega\left(\frac{n}{D}\right)
\end{eqnarray*}

For the case where $n \leq 64D$, we approximate:

\begin{equation*}
	\left(1-\frac{1}{64D}\right)^n \leq \left(1-\frac{1}{64D}\right)^{64D} \approx \frac{1}{e} < \frac{1}{2}
\end{equation*}

Therefore, we conclude that $(1-1/(64D))^n \leq \max\{1-\Omega(n/D),1/2\}$.

\end{proof}

\begin{theorem}
\label{thm:upper}
Let each of $n$ agents execute a copy of Algorithm~\ref{algo:non-uni-loop}. The minimum over all agents of the expected number of moves of an agent to find a target within distance $D > 1$ from the origin is $\BO(D^2/n + D)$.
\end{theorem}
\begin{proof}
For $i\in \N$, denote by ${\cal E}_i$ the event that some agent finds the target in iteration $i$ of the main loop, but no agent finds it in any previous iteration $i' < i$. 
Denote by $q$ the probability that no agent finds the target in iteration $i$ of the main loop. 
Note that the events ${\cal E}_i$ are mutually exclusive. From the independence of random choices of different agents and iterations, it thus follows that
\begin{equation*}
P\left[{\cal E}_i\right] = (1-q)q^{i-1}.
\end{equation*}

Let random variable $X_{\text{found}}$ denote the number of moves until the first agent finds the target.
\begin{equation*}
\E[X_{\text{found}}] = \sum_{i=1}^{\infty} P[{\cal E}_i] \cdot \E[X_{\text{found}}| {\cal E}_i] 
\end{equation*}

We partition event ${\cal E}_i$ into disjoint events ${\cal E}_{i,a}$, where ${\cal E}_{i,a}$ denotes the event that agent $a$ is the agent with the smallest id that finds the target in iteration $i$\footnote{We pick the agent with the smallest id just as a tie-breaker between agents; the ids of agents do not play an important role in the algorithm.}. By the definition of ${\cal E}_i$, we know that in any execution in ${\cal E}_i$, some agent finds the target in iteration $i$.

Next, we bound the value of $\E[X_{\text{found}}| {\cal E}_i] $. By the Law of Total Expectation applied to the partition of event ${\cal E}_i$, it follows:

\begin{equation*}
\E[X_{\text{found}}| {\cal E}_i] = \sum_{a} P[{\cal E}_{i,a} | {\cal E}_i] \cdot \E[X_{\text{found}} | {\cal E}_{i,a}] 
\end{equation*}

Let random variable $X_{i,a}$ denote the number of moves an agent $a$ takes to complete iteration $i$. Note that because we condition on event ${\cal E}_{i,a}$, we know that the expected number of rounds for some agent to find the target is at least the expected number of rounds for the fixed agent $a$ to find the target and to complete iteration $i$. Therefore, it follows that: 

\begin{equation*}
\E[X_{\text{found}}| {\cal E}_i] \leq \sum_{a} P[{\cal E}_{i,a} | {\cal E}_i] \cdot \E[X_{i,a} | {\cal E}_{i,a}]
\end{equation*}

By the definition of $R_{i,a}$, we know that $\E[X_{i,a} | {\cal E}_{i,a}] = R_{i,a}$ because the value of $R_{i,a}$ is the same for any fixed agent $a$, including the one with the smallest id. Using Lemma \ref{lem:ria}, we conclude that

\begin{equation*}
	\E[X_{\text{found}}| {\cal E}_i] \leq \sum_{a} P[{\cal E}_{i,a} | {\cal E}_i] \cdot R_{i,a} \leq \sum_{a} P[{\cal E}_{i,a} | {\cal E}_i] \cdot 4iD = 4iD \cdot \sum_{a} P[{\cal E}_{i,a} | {\cal E}_i] = 4iD
\end{equation*}

Finally, we sum over all iterations to calculate the value of $\E[X_{\text{found}}]$. Recall that we already calculated the value of $P[{\cal E}_i]$.
\begin{eqnarray*}
\E[X_{\text{found}}] &\leq & \sum_{i=1}^{\infty} P[{\cal E}_i] \cdot \E[X_{\text{found}}| {\cal E}_i] \leq \sum_{i=1}^{\infty}(1-q)q^{i-1}\cdot 4iD \\
 & \leq &  4D\sum_{i=0}^{\infty}(1-q)q^i i = 4D\cdot\frac{q}{1-q}\leq \frac{4D}{1-q}.
\end{eqnarray*}
By Lemma~\ref{lem:iterations}, we know that $q\leq\max\{1-\Omega(n/D),1/2\}$, so it follows that
\begin{equation*}
\E[X_{\text{found}}] \leq \frac{4D}{1-q}= \BO\left(\max\left\{\frac{D^2}{n},D\right\}\right)=\BO\left(\frac{D^2}{n}+D\right).\qedhere
\end{equation*}
\end{proof}

We now generalize this algorithm to one that uses probabilities lower bounded by $1/2^\ell$ for some given $\ell\geq 1$. This is achieved by the following subroutine, which implements a coin that shows tails with probability $1/2^{k\ell}$ using a biased coin that shows tails with probability $1/2^\ell$, for $\ell \geq 1$. 

\begin{algorithm}
\caption{coin($k,\ell$): Biased coin flip showing tails with probability $1/2^{k\ell}$.}
\label{algo:p_coin}
\For{$i = 0 \cdots k$} {
	\If{$C_{1/2^\ell}$ shows heads}
	{
	\Return{heads}\phantom{Xp}
	}
}
\Return{tails}\phantom{Xp}\nllabel{line:p_tails}
\end{algorithm}

\begin{lemma}
\label{lem:p_coin}
	Algorithm~\ref{algo:p_coin} returns tails with probability $1/2^{k\ell}$ and requires $\lceil\log k\rceil$ bits of memory.
\end{lemma}
\begin{proof}
	From the code it follows that the action on Line~\ref{line:p_tails} is performed only if none of the outcomes of the coin flips are tails. Since each coin shows tails with probability $1/2^\ell$ and there is a total if $k$ coin flips, the probability of all of them being tails is $1/2^{k\ell}$. Since the entire state of the algorithm is the loop counter, it can be implemented using $\lceil\log k\rceil$ bits of memory.
\end{proof}

Next, we show how to combine Algorithm \ref{algo:non-uni-loop} and Algorithm \ref{algo:p_coin}, and we analyze the performance and selection complexity of the resulting algorithm. Given a biased coin $C_{1/2^{\ell}}$, we construct Algorithm Non-Uniform-Search  by replacing the lines where coin $C_{1/D}$ is tossed in Algorithm \ref{algo:non-uni-loop} with a copy of Algorithm \ref{algo:p_coin}, with parameters $k = \lceil \log D / \ell \rceil$ and $\ell$.

\begin{theorem}
\label{thm:algoA}
Let each of $n$ agents execute a copy of Algorithm Non-Uniform-Search. The minimum over all agents of the expected number of moves for an agent to find a target in distance $D > 1$ from the origin is $\BO(D^2/n + D)$. Moreover, Algorithm Non-Uniform-Search satisfies $\chi(\text{Algorithm Non-Uniform-Search}) = \log \log D +\BO(1)$.
\end{theorem}

\begin{proof}
	 By Lemma \ref{lem:p_coin}, Algorithm \ref{algo:p_coin} run with parameters $k = \lceil \log D / \ell \rceil$ and $\ell$ generates coin flips with probability $1/D$ of showing tails. Therefore, the correctness of Algorithm Non-Uniform-Search follows from Theorem \ref{thm:upper}. Since Algorithm \ref{algo:p_coin} does not generate any moves of the agents on the grid, the time complexity of algorithm also follows from Theorem \ref{thm:upper}. 
	 
	 Finally, by Lemma \ref{lem:p_coin} and the fact that Algorithm \ref{algo:non-uni-loop} requires $3$ bits to be implemented, it follows that $\chi(\text{Algorithm Non-Uniform-Search}) = b + \log \ell = \log \lceil \log D / \ell \rceil + \log \ell + 3 = \log \log D + \BO(1)$.
\end{proof}

\subsection{Uniform Algorithm}
\label{sec:uniform}

In this section, we generalize the results from Section \ref{sec:non-uniform} to derive an algorithm that is uniform in $D$. The main difference is that now each agent maintains an estimate of $D$ that is increased until the target is found. For each estimate, an agent simply executes the corresponding variant of Algorithm Non-Uniform-Search. We show that for the algorithm in this section, the expected number of moves for the first agent to find a target at distance at most $D$ from the origin is $(D^2/n + D) 2^{\BO(\ell)}$. Also, the algorithm uses only $b = 3 \log \log_{2^\ell} D +\BO(1) = 3 \log \log D - 3 \log \ell + \BO(1)$ bits of memory.

To simplify the presentation, we break up the main algorithm into subroutines. We begin by showing how to move in a given direction by a random number of moves that depends on the current estimate $\hat{D}$ of $D$. In the following algorithm, recall that $\ell$ is used to bound from below the smallest probability available to each agent by $1/2^{\ell}$. We use an integer $k$ as a parameter to the algorithm in order to generate different distance estimates $\hat{D} = 2^{k \ell}$.

\begin{algorithm}
\caption{walk($k$,$\ell$, $dir$): Move by a random number of moves in direction $dir$ that is roughly uniform on $0,\ldots,2^{k\ell}$.}
\label{algo:goToPoint}
\While{coin($k,\ell$)$\,=\,$heads}{
  move one step in direction $dir$
}
\end{algorithm}

\begin{lemma}
\label{lem:walk}
	For each $i\in \{0,\ldots,2^{k\ell}\}$, the probability that Algorithm~\ref{algo:goToPoint} performs exactly $i$ moves is at least $1/2^{k\ell+2}$. The probability that the algorithm performs at least $2^{k\ell}$ moves is at least $1/4$. The expected number of moves is smaller than $2^{k\ell}$. The algorithm requires $\lceil \log k\rceil$ bits of memory.  
\end{lemma}

\begin{proof}
By Lemma~\ref{lem:p_coin}, the probability that the algorithm performs exactly $i\leq 2^{k\ell}$ moves is
\begin{equation*}
\left(1-\frac{1}{2^{k\ell}}\right)^i\frac{1}{2^{k\ell}}\geq \left(1-\frac{1}{2^{k\ell}}\right)^{2^{k\ell}}\frac{1}{2^{k\ell}}\geq \frac{1}{2^{k\ell+2}}.
\end{equation*}
The probability that it performs at least $2^{k\ell}$ moves is $(1-1/2^{k\ell})^{2^{k\ell}}\geq 1/4$. The expected number of moves is
\begin{equation*}
\sum_{i=1}^{\infty} i \left(1-\frac{1}{2^{k\ell}}\right)^i\frac{1}{2^{k\ell}}=\frac{1-1/2^{k\ell}}{(1/2^{k\ell})^2}\cdot \frac{1}{2^{k\ell}}<2^{k\ell}.
\end{equation*}
\noindent Implementing the coin flip by Algorithm~\ref{algo:p_coin}, the memory requirement follows from Lemma~\ref{lem:p_coin}.
\end{proof}

Using the subroutines above, Algorithm~\ref{algo:search} visits each grid point of a square of side length $2^{k\ell}$ centered at the origin with probability $\Omega(1/2^{2k\ell})$.

 \begin{algorithm}
\caption{search($k,\ell$): Visit each grid point of a square of side length $2^{k\ell}$ centered at the origin with probability $\Omega(1/2^{2k\ell})$.}
\label{algo:search}
\If{if $C_{1/2}$ shows heads}{
  walk($k,\ell$,up)
}
\Else{
  walk($k,\ell$,down)
}
\If{$C_{1/2}$ shows heads}{
  walk($k,\ell$,right)
}
\Else{
  walk($k,\ell$,left)
}
\end{algorithm}

\begin{lemma}
\label{lem:uniformArea}
  If called at the origin, for each point $(x,y) \in \{0,\ldots,2^{k\ell}\}^2$, Algorithm~\ref{algo:search} visits point $(x,y)$ with probability at least $1/2^{k\ell+6}$. It can be implemented using $\lceil \log k \rceil+2$ bits of memory.
\end{lemma}
\begin{proof}
	Consider grid point $(x,y)\in \{0,\ldots,2^{k\ell}\}^2$. With probability $1/2$ each, the algorithm decides to move up and right in the first and second call to Algorithm~\ref{algo:goToPoint}. By Lemma~\ref{lem:walk}, the first call will halt after exactly $x$ moves with probability at least $1/2^{k\ell+2}$, and the second call will perform $2^{k\ell}\geq y$ moves with probability at least $1/4$. Hence, the claimed lower bound on the probability to visit $(x,y)$ follows. Analogously, the same holds for $(-x,y)$, $(-x,-y)$, and $(x,-y)$. The memory requirements are $2$ bits to memorize whether the direction of movement is currently up, down, left, or right, plus the $\lceil \log k \rceil$ bits needed for the (sequential) calls to Algorithm~\ref{algo:goToPoint}.
\end{proof}

Finally, in Algorithm~\ref{algo:loop}, we use Algorithm \ref{algo:search} to efficiently search an area of $\BO(D^2)$ with $n$ agents. Intuitively, the algorithm iterates through different values of the outer-loop parameter $i$, which correspond to the different estimates of $D$, increasing by approximately a factor of $2^{\ell}$. For each such estimate, the algorithm needs to execute a number of calls to the search subroutine with parameter $i$. 
However, since agents have limited memory and limited probability values, we can only count the number of such calls to the search routine approximately. We do so similarly to Algorithm \ref{algo:goToPoint}, by repeatedly tossing a biased coin and calling the search algorithm as long as the coin shows heads. 

\begin{algorithm}
\caption{Search Algorithm for $n$ agents. $K$ is a sufficiently large constant.}
\label{algo:loop}
\For{$i=1,\ldots$}{
  \While{coin($K + \max\{i-\lfloor (\log n)/\ell\rfloor,0\},\ell$)$\,=\,$heads}{
    search($i,\ell$)\\
    return to the origin
  }
}
\end{algorithm}

Throughout the proof of Algorithm \ref{algo:loop}, we refer to an iteration of the outer-most loop as a phase. 

\paragraph{Proof Overview.} First, in Lemma \ref{lem:exp_i}, we calculate the expected number of moves $R_i$ for an agent to complete phase $i$. Then, we apply the same reasoning as in Lemma \ref{lem:condition} (from Section \ref{sec:non-uniform}), to determine the expected number of  moves $\tilde{R}_{i,a}$ for an agent $a$ to complete phase $i$ (past some initial number of $\lceil \log_{2^{\ell}} D \rceil$ phases), conditioning on agent $a$ finding the target in phase $i$. Next, we move on to reasoning about all $n$ agents, instead of a single agent. In Lemma \ref{lem:whileLoop}, we bound the probability that in each phase $i$, at least $\Omega(2^{i\ell})$ calls to the subroutine search$(i,\ell)$ are executed by all agents together. In Lemma \ref{lem:probFind}, we use that result to calculate the probability that at least one of the $n$ agents finds the target in some phase $i$. Finally, we use these intermediate results to prove the main result of this section, Theorem \ref{theorem:upper}, which shows that the expected number of moves for the first agent to find a target within distance $D$ from the origin is $2^{\BO(\ell)}(D+D^2/n)$.

Denote by $R_i$ the expected number of moves until an agent completes phase $i$. 

\begin{lemma}\label{lem:exp_i}
$R_i\leq 4\rho_i 2^{i\ell}$, where $\rho_i:=2^{(K+\max\{i-\lfloor (\log n)/\ell\rfloor,0\})\ell}$.
\end{lemma}
\begin{proof}
By linearity of expectation, we can calculate $R_i$ as follows. In the expression below, index $i'$ counts the number of phases (up to $i$), index $j$ counts the number of calls to the search subroutine, and index $k$ counts the number of moves for an agent to complete each call to the search subroutine. In the sum indexed by $j$, we just calculate the probability that exactly $j$ calls to the search subroutine are executed and multiply that by the expected number of moves to complete one such call. In the sum indexed by $k$, we calculate the probability that Algorithm \ref{algo:goToPoint} stops after $k$ moves and multiply that by $2k$ because  each call to search($i',\ell$) results in two calls to walk($i',\ell,\cdot$). 
\begin{eqnarray*}
R_i &=&\sum_{i'=1}^{i} \left( \sum_{j=0}^{\infty} \frac{1}{\rho_{i'}} \left(1-\frac{1}{\rho_{i'}}\right)^j \sum_{k=0}^{\infty}\frac{1}{2^{i'\ell}} \left(1-\frac{1}{2^{i'\ell}}\right)^k\cdot 2k \right) \\
&<& \sum_{i'=1}^{i} \left( \sum_{j=0}^{\infty} \frac{1}{\rho_{i'}} \left(1-\frac{1}{\rho_{i'}}\right)^j 2\cdot 2^{i'\ell} \right)\\
&<& \sum_{i'=1}^{i}  2 \rho_{i'} 2^{i'\ell}\\
&<& 4\rho_i 2^{i\ell},
\end{eqnarray*}
\end{proof}

Denote by $\tilde{R}_{i,a}$ the expected number of moves until a fixed agent $a$ finds the target, conditioning on the fact that agent $a$ finds the target in phase $i\in \N$, but no earlier phase.

\begin{corollary}
\label{coro:time}
	If $i\geq i_0 =\lceil \log_{2^\ell} D\rceil$, then it holds that $\tilde{R}_{i,a} \leq 8\rho_i 2^{i\ell}=2^{\max\{2i\ell-\log n,i\ell\}}\cdot 2^{\BO(\ell)}$.
\end{corollary}
\begin{proof}
In the last phase, agent $a$ walks directly to the target, which takes at most $2D$ moves. For all previous phases, reasoning analogously to Lemma~\ref{lem:condition}, in terms of phases instead of iterations, we see that the expectation does not increase by more than a factor of $2$ due to conditioning on not finding the target. The claim thus follows from Lemma~\ref{lem:exp_i} and the fact that $2^{i\ell}\geq 2^{i_0\ell}\geq D$.
\end{proof}

Denote by ${\cal E}_1(i)$ the event that in total at least $2^{(K/2+i)\ell}$ calls to search$(i,\ell)$ are executed in phase $i$.

\begin{lemma}
\label{lem:whileLoop}
	$P[{\cal E}_1(i)]\geq 1-1/2^{2\ell+2}$.
\end{lemma}
\begin{proof}
Abbreviate $\rho_i :=2^{(K + \max\{i-\lfloor (\log n)/\ell\rfloor,0\})\ell}\geq 2$. By Lemma~\ref{lem:p_coin} and linearity of expectation, the expected number of calls to search($\ell,i$) performed by all agents during phase $i$ is
\begin{equation*}
n\sum_{j=1}^{\infty}\frac{1}{\rho_i}\cdot\left(1-\frac{1}{\rho_i}\right)^j\cdot j
= n\cdot \rho_i \left(1-\frac{1}{\rho_i}\right)\geq  \frac{n}{2}\cdot \rho_i
\geq 2^{(K + i-1)\ell}.
\end{equation*}

Since the coin flips are independent, we can apply Chernoff's bound (see Equation \eqref{eq:chernoff_lower} in the Appendix), showing that the probability that fewer than $2^{(K/2 + i)\ell}$ searches are executed in total is at most $e^{-\Omega(K\ell)}$. Hence, since $K$ is a sufficiently large constant, the claim follows.
\end{proof}

Denote by ${\cal E}_2(i)$ the event that the target is found by some agent in phase $i\geq  i_0$.

\begin{lemma}
\label{lem:probFind}
	 $P[{\cal E}_2(i)]\geq 1-1/2^{2\ell+1}$.
\end{lemma}
\begin{proof}
By Lemma~\ref{lem:whileLoop}, with probability at least $1-1/2^{2\ell+2}$, at least $2^{(K/2+i)\ell}$ iterations of the while loop are executed in total. Because $i\geq i_0\geq \log_{2^\ell}D$, i.e., $2^{i\ell}\geq D$, Lemma~\ref{lem:uniformArea} shows that in each iteration, the probability to find the target is at least $1/2^{i\ell+6}$. Therefore, the probability to miss the target in all calls is at most
\begin{equation*}
\left(1-\frac{1}{2^{i\ell+6}}\right)^{2^{(K/2+i)\ell}}=2^{-\Omega(K\ell)}.
\end{equation*}
Because $K$ is a sufficiently large constant, we may assume that this is at most $1/2^{2\ell+2}$. We conclude that
\begin{equation*}
P[{\cal E}_2(i)]\geq P[{\cal E}_2(i)\,|\,{\cal E}_1(i)]\cdot P[{\cal E}_1(i)]\geq \left(1-\frac{1}{2^{2\ell+2}}\right)^2\geq 1-\frac{1}{2^{2\ell+1}},
\end{equation*}
as claimed.
\end{proof}

Let event ${\cal E}_3(i)$ denote the event that the target is found for the first time in phase $i$.

\begin{theorem}\label{theorem:upper}
Let each of $n$ agents execute a copy of Algorithm~\ref{algo:loop}. The minimum over all agents of the expected number of moves for an agent to find a target within distance $D$ from the origin is $2^{\BO(\ell)}(D+D^2/n)$.
\end{theorem}
\begin{proof}
Observe that because the probability to find the target in phase $i$ is independent of all coin flips in earlier phases, we have that
\begin{equation*}
P[{\cal E}_3(i)]=P[{\cal E}_2(i)]\prod_{i'=1}^{i-1}(1-P[{\cal E}_2(i')]).
\end{equation*}
For $i> i_0$, by Lemma~\ref{lem:probFind} it follows that
\begin{equation*}
P[{\cal E}_3(i)]\leq \prod_{i'=i_0}^{i-1}(1-P[{\cal E}_2(i')])\leq \frac{1}{2^{(2\ell+1)(i-i_0)}}.
\end{equation*}

Let random variable $X_{\text{found}}$ denote the number of moves until the first agent finds the target.

\begin{equation}
\E[X_{\text{found}}] = \sum_{i=1}^{\infty} P[{\cal E}_3(i)]\cdot \E[X_{\text{found}} | {\cal E}_3(i)]
\end{equation}

We partition event ${\cal E}_3(i)$ into disjoint events ${\cal E}_3(i,a)$, where ${\cal E}_3(i,a)$ denotes the event that agent $a$ is the agent with the smallest id that finds the target in phase $i$\footnote{We pick the agent with the smallest id just as a tie-breaker between agents; the ids of agents do not play an important role in the algorithm.}. By the definition of ${\cal E}_3(i)$, we know that in any execution in ${\cal E}_3(i)$, some agent finds the target in phase $i$.

Next, we bound the value of $\E[X_{\text{found}}| {\cal E}_3(i)] $. By the Law of Total Expectation applied to the partition of event ${\cal E}_3(i)$, it follows:

\begin{equation*}
\E[X_{\text{found}}| {\cal E}_3(i)] = \sum_{a} P[{\cal E}_3(i,a) | {\cal E}_3(i)] \cdot \E[X_{\text{found}} | {\cal E}_3(i,a)] 
\end{equation*}

Let random variable $X_{i,a}$ denote the number of moves an agent $a$ takes to complete iteration $i$. Note that because we condition on event ${\cal E}_3(i,a)$, we know that the expected number of rounds for some agent to find the target is at least the expected number of rounds for the fixed agent $a$ to find the target and complete iteration $i$. Therefore, it follows that: 

\begin{equation*}
\E[X_{\text{found}}| {\cal E}_3(i)] \leq \sum_{a} P[{\cal E}_3(i,a) | {\cal E}_3(i)] \cdot \E[X_{i,a} | {\cal E}_3(i,a)]
\end{equation*}

By the definition of $\tilde{R}_{i,a}$, we know that $\E[X_{i,a} | {\cal E}_3(i,a)] = \tilde{R}_{i,a}$ because the value of $\tilde{R}_{i,a}$ is the same for each fixed agent $a$, including the one with the smallest id. Therefore, we conclude that:
\begin{equation*}
	\E[X_{\text{found}}| {\cal E}_3(i)] \leq  \sum_{a} P[{\cal E}_3(i,a) | {\cal E}_3(i)] \cdot \tilde{R}_{i,a} = \tilde{R}_{i,a}  \cdot \sum_{a} P[{\cal E}_3(i,a) | {\cal E}_3(i)] = \tilde{R}_{i,a} 
\end{equation*}

Finally, we sum over all phases to calculate the value of $\E[X_{\text{found}}]$. Recall that we already calculated the value of $P[{\cal E}_3(i)]$. Using Corollary \ref{coro:time} for the value of $\tilde{R}_{i,a}$, we conclude that:
\begin{eqnarray*}
\E[X_{\text{found}}] &=& \sum_{i=1}^{i_0} P[{\cal E}_3(i)]\cdot \tilde{R}_{i,a} + \sum_{i=i_0+1}^{\infty} P[{\cal E}_3(i)]\cdot \tilde{R}_{i,a} \\
&\leq & \tilde{R}_{i_0,a} + \sum_{i=i_0+1}^{\infty} \frac{1}{2^{(2\ell+1)(i-i_0)}}\cdot 2^{\max\{2i\ell-\log n,i\ell\}}\cdot 2^{\BO(\ell)}\\
&\leq &  2^{\max\{2i_0\ell-\log n,i_0\ell\}}\cdot 2^{\BO(\ell)}\cdot \sum_{i=i_0}^{\infty} \frac{2^{2\ell(i-i_0)}}{2^{(2\ell+1)(i-i_0)}} \\
&\leq & \max\left\{\frac{D^2}{n},D\right\}\cdot 2^{\BO(\ell)}\cdot \sum_{i=i_0}^{\infty} 2^{-(i-i_0)}\\
&\leq & \max\left\{\frac{D^2}{n},D\right\}\cdot 2^{\BO(\ell)}.
\end{eqnarray*}
\vspace*{-1.35cm}

\qedhere
\end{proof}

\section{Lower bound}
\label{sec:lower}

In this section, we present a lower bound showing that there is no algorithm that finds a target placed within distance $D$ from the origin in $D^{2-o(1)}$ rounds with high probability (w.h.p.), such that the algorithm satisfies $\chi({\cal A}) \leq \log \log D - \omega(1)$. 

Throughout this section, we say that some event occurs with high probability iff the probability of the event occurring is at least $1 - 1/D^c$ for an arbitrary predefined constant $c > 0$ and some $D \in \N$. We say that two probability distributions $\pi_1$ and $\pi_2$ are ``approximately equivalent'' iff $\|\pi_1 - \pi_2\| = \BO(1/D^c)$ for an arbitrary predefined constant $c > 0$ and some $D \in \N$. By $\|\cdot\|$ we denote the $\infty$-norm on the respective space.

First, we state the main theorem of the section in terms of the performance metric $M_{\text{steps}}$ (the minimum over all agents of the number of steps for an agent to find the target). Note that, by the definition of a round, this is equivalent to counting the expected number of rounds until the first agent finds the target. At the end of the section, in Corollary \ref{cor:moves}, we generalize the main result to apply to metric $M_{\text{moves}}$ (the minimum over all agents of the number of moves for an agent to find the target). 

\begin{theorem}\label{thm:lower}
	Let $\mathcal{A}$ be an algorithm with $\chi({\cal A}) = b + \log \ell \leq \log \log D-\omega(1)$ and $n \in poly(D)$ agents. There is a placement of the target within distance $D$ from the origin such that w.h.p.\ no agent executing algorithm $\mathcal{A}$ finds it in fewer than $D^{2-o(1)}$ rounds.
	Moreover, the probability for some agent to find a target, placed uniformly at random in the square of side $2D$ centered at the origin, within $D^{2-o(1)}$ rounds is $o(1)$.
\end{theorem}

\subsection{Proof Overview}

Here we provide a high-level overview of our main proof argument. We fix an algorithm $\mathcal{A}$ and focus on executions of this algorithm of length $D^{2-o(1)}$ rounds. We prove that since agents have $o(\log D)$ states, they ``forget'' about past events too fast to behave substantially different from a biased random walk. 

More concretely, first we show, in Corollary \ref{cor:initial} that after $D^{o(1)}$ initial rounds each agent $a$ is located in some recurrent class $C(a)$ of the Markov chain. We use this corollary to prove, in Corollary \ref{cor:origin}, that after the initial $D^{o(1)}$ rounds each agent $a$ does not return to the origin (or it keeps returning every $D^{o(1)}$ rounds, so it does not explore much of the grid). Therefore, throughout the rest of the proof we can ignore the states labeled ``origin".  

Assume there is a unique stationary distribution of $C(a)$.\footnote{This holds only if the induced Markov chain on the recurrent class is aperiodic, but the reasoning is essentially the same for the general case. We handle this technicality in Section \ref{sec:draw}.} Since there are few states and non-zero transition probabilities are bounded from below, standard results on Markov chains imply that taking $D^{o(1)}$ steps from any state in the recurrent class will result in a distribution on the class's states that is (almost) indistinguishable from the stationary distribution (Corollary~\ref{cor:bound}); in other words, any information agents try to preserve in their state will be lost quickly with respect to $D$.

The next step in the proof is a coupling argument. We split up the rounds in the execution into groups such that within each group, rounds are sufficiently far apart from one another for the above ``forgetting'' to take place. For each group, we show that drawing states independently from the stationary distribution introduces only a negligible error (Lemma~\ref{lem:pi_s} and Corollary~\ref{cor:coin}). Doing so, we can apply Chernoff's bound to each group, yielding that agents will not deviate substantially from the expected path they take when, in each round, they draw a state according to the stationary distribution and execute the corresponding move on the grid (Lemma~\ref{lem:upwards} and Corollary~\ref{cor:concentration}). Taking a union bound over all groups, it follows that, w.h.p., each agent will not deviate from a straight line (the expected path associated with the recurrent class it ends up in) by more than distance $o(D/|S|)$, where $S$ is the number of states of the Markov chain. It is crucial here that the corresponding region in the grid, restricted to distance $D$ from the origin, has size $o(D^2/|S|)$ and depends only on the component of the Markov chain the agent ends up in. Therefore, since there are no more than $|S|$ components, taking a union bound over all agents shows that w.h.p.\ together they visit an area of $o(D^2)$.

\subsection{Proof}
We assume without loss of generality that values like $\ln D$ are integers; for the general case, one may simply round up. Moreover, since we are interested in asymptotics with respect to $D$, we may always assume that $D$ is larger than any given constant.

Fix any algorithm ${\cal A}$, some $D \in \N$, and let $b + \log \ell \leq \log \log D-\omega(1)$. Consider the probability distribution of executions of ${\cal A}$ of length $\Delta=D^{2-o(1)}$ rounds; we will fix the $o(1)$-term in the exponent later, in Lemma~\ref{lem:upwards}. 

We break the proof down into three main parts. First, in Section \ref{sec:initial}, we show that after a certain number of initial rounds each agent is in a recurrent class and, for simplicity, we can ignore the states labeled ``origin". Next, in Section \ref{sec:draw}, we show that if we break down the execution into large enough blocks of rounds, with high probability  we can assume that the steps associated with rounds in different blocks do not depend on each other. Finally, in Section \ref{sec:movement}, we focus on the movement of the agents in the grid, derived from these ``almost" independent steps, and we show that with high probability the agents will not explore any points outside of an area of size $o(D)$ around the origin.

\subsubsection{Initial steps in the Markov chain}
\label{sec:initial}

Let random variable $C(a,r)$ denote the recurrent class of the Markov chain in which agent $a$ is located at the end of round $r$; if $a$ is in a transient state at the end of round $r$, we set $C(a,r):=\bot$. Also, by $p_0$ we denote the smallest non-zero probability in the Markov chain. By assumption, we know that $p_0 \geq 1/2^\ell$.

First we show that for any agent $a$ and any state $s$ of the Markov chain, if state $s$ is always reachable by agent $a$, then agent $a$ visits state $s$ within $D^{o(1)}$ rounds.

Let $R_0 = p_0^{-2^b}2^b c\log D = D^{o(1)}$ where the constant $c>0$ will be specified later.

\begin{lemma}
\label{lem:reach}
For any agent $a$, any round $r$, and any state $s$, condition on the event that at the end of round $r$ of the execution agent $a$ never visits a state $s'$ such that $s$ is not reachable from $s'$. Then, w.h.p.\ agent $a$ visits state $s$ at the end of round $r'$ such that $r \leq r' \leq r+R_0$.
\end{lemma}
\begin{proof}
	Since state $s$ is reachable after each round, there must be a path of length at most $|S|-1$ from the state in which agent $a$ is located at the end of round $r$ to state $s$. Therefore, the probability that the agent visits state $s$ within $R_0$ rounds is bounded from below by the probability that a potentially biased random walk on a line of $2^b\geq |S|-1$ nodes starting at the leftmost node reaches the rightmost node within $R_0$ rounds. This probability in turn is bounded from below by
\begin{equation*}
1 - \left(1-p_0^{2^b}\right)^{R_0/2^b}=1 - \left(1-p_0^{2^b}\right)^{p_0^{-2^b}c\log D}
\geq 1 - 1/2^{\Omega(c\log(D+n))}=1 - \frac{1}{D^{\Omega(c)}},
\end{equation*}
where the second last step uses fact that $p_0 \leq 1/2$; otherwise, each state in the Markov chain has only one outgoing transition, in which case the claim of the lemma is trivial.
Therefore, for an appropriate choice of $c$, w.h.p.\ agent $a$ visits state $s$ within $R_0$ rounds.
\end{proof}

In the following corollary we show that for any round $r \geq R_0$, w.h.p.\ an agent is located in some recurrent class of the Markov chain.

\begin{corollary}
\label{cor:initial}
	For any agent $a$ and any round $r \geq R_0$, w.h.p.\ $C(a,r) = C(a,r+1) \neq \bot$.
\end{corollary}
\begin{proof}
First, we derive a Markov chain from the original Markov chain as follows. We identify all recurrent states in the original Markov chain and we merge them all into a single recurrent state $s_C$ of the derived Markov chain.

By definition of a recurrent class and because there is only one such class, for each state $s$ in the derived Markov chain, the recurrent state $s_C$ is always reachable from $s$. Applying Lemma \ref{lem:reach} to agent $a$, round $0$ and state $s_C$, it follows that w.h.p.\ $a$ visits $s_C$ at the end of $R_0$ or earlier. This implies that in the original Markov chain w.h.p.\ agent $a$ visits \emph{some} recurrent state $s \in C$, where $C$ is a recurrent class, at the end of round $R_0$ or earlier.
Since recurrent classes cannot be left, it must be the case that $C(a,R_0) = C$ and $C(a,r) = C(a,r+1) = C$ for each $r \geq R_0$.
\end{proof}

In Corollary \ref{cor:initial}, we showed that at the end of each round $r \geq R_0$ an agent $a$ is located in some recurrent class $C(a,r)$ w.h.p.\ Since agent $a$ does not leave that class in subsequent rounds, we will refer to it by $C(a)$. Next, we show that for any state $s \in C(a)$ and any round $r \geq R_0$, w.h.p.\ agent $a$ visits state $s$ at the end of round $r+R_0$ or earlier.

\begin{corollary}
\label{cor:visits}
For each agent $a$, each round $r \geq R_0$, and each state $s \in C(a)$, w.h.p.\ agent $a$ visits state $s$ at the end of round $r'$ such that $r \leq r' \leq r+R_0$.
\end{corollary}

\begin{proof}
	By Corollary \ref{cor:initial}, we know that w.h.p.\ at the end of round $r \geq R_0$ agent $a$ is in some recurrent class $C(a)$ that it never subsequently leaves. Therefore, each state $s \in C(a)$ is reachable by $a$ after each round $r \geq R_0$. By Lemma \ref{lem:reach},  it is true that w.h.p.\ agent $a$ visits state $s$ within $R_0$ rounds.
\end{proof}

Finally, we show that the recurrent class $C(a)$ in which agent $a$ is located does not contain any states labeled ``origin'', or otherwise, the agent keeps returning to the origin too often and makes no progress exploring the grid. This result will let us, for convenience, ignore states labeled ``origin" through the rest of the proof.

\begin{corollary}
\label{cor:origin}
W.h.p., at least one of the following is true for any agent $a$ in any round $r \geq R_0$: (1) agent $a$ never visits a point in the grid at distance more than $D^{o(1)}$ from the origin, or (2) agent $a$ is located in a recurrent class in which none of the states are labeled ``origin".
\end{corollary}
\begin{proof}
	By Corollary \ref{cor:initial}, we know that w.h.p.\ at the end of round $r \geq R_0$ agent $a$ is in some recurrent class $C(a)$ that it never subsequently leaves. By Lemma \ref{cor:visits}, we know that for each state $s \in C(a)$, labeled ``origin" and each round $r \geq R_0$, w.h.p.\ agent $a$ visits $s$ at the end of round $r+R_0$ or earlier. Therefore, if $C(a)$ contains a state labeled ``origin", then w.h.p.\ agent $a$ never visits a point at distance more than $R_0 = D^{o(1)}$ from the origin. Otherwise, (2) holds.
\end{proof}

Throughout the rest of the proof, we consider executions after round $R_0$; since, $R_0 = D^{o(1)}$ and we consider executions of length $\Delta = D^{2-o(1)}$, we can just ignore these initial rounds. Therefore, from Corollary \ref{cor:initial} and Corollary \ref{cor:origin}, we can assume for the rest of the proof that each agent $a$ is in a recurrent class $C(a)$ and it does not return to the origin.

\subsubsection{Moves drawn from the stationary distribution}
\label{sec:draw}

Fix an agent $a$ and consider $C:=C(a)$. As $|C|\leq |S|\in o(\log D)$, the Markov chain induced by $C$ has a period of $t = o(\log D)$ (an aperiodic chain has period $t=1$). We apply Theorem~\ref{thm:feller} to $C$ and denote by $G_1,\ldots,G_t$ the equivalence classes based on the period $t$ whose existence is guaranteed by the theorem.

Consider blocks of rounds of size $\beta = c |S| \ln D / p_0^{|S|}=D^{o(1)}$, where $c>0$ is a constant that is determined by Corollary~\ref{cor:bound}. Without loss of generality, assume that $\beta$ is a multiple of $t$ (otherwise use $t^{\lceil\log_t \beta\rceil}=\BO(\beta)$ instead). We define group of rounds such that each group contains one round from each block. Formally, for $1 \leq i \leq \beta$ and $j \in \N_0$, group $B_i$ contains round numbers $i + j\beta \leq \Delta$. Observe that this definition entails that at the end of all rounds from a given group, agent $a$ is in some state from the same class $G_{\tau}\subseteq C$ that is recurrent and closed under $P^t$.

Let $\pi_{r+\beta,s}$ denote the probability distribution on $G_{\tau}$ of possible states agent $a$ may be in at the end of round $r+\beta$, conditional on its state being $s\in G_{\tau}$ at the end of round $r$. Note that this distribution is, in fact, independent of $r$. We obtain the following corollary of Lemma~\ref{lem:rosenthal} applied to the Markov chain induced by the matrix $P^t$ restricted to class $G_{\tau}$.

\begin{corollary}
\label{cor:bound}
	There is a unique stationary distribution $\pi_{\tau}$ of the Markov chain on $G_{\tau}$ induced by $P^t$. For any state $s \in G_{\tau}$ and any round $r$, $\pi_{r+\beta,s}$ and $\pi_{\tau}$ are approximately equivalent.
\end{corollary}

\begin{proof}
	Since $\beta$ is a multiple of $t$, we can consider the probability matrix $P^t$, which by Theorem~\ref{thm:feller} induces a Markov chain on $G_{\tau}$. We apply Lemma~\ref{lem:rosenthal} to this chain with the following parameters: $k_0 = |S|/t$,\footnote{Recall that we assume such values to be integer, otherwise rounding is required.} $Q(s) = 1$ (i.e., $Q(s')=0$ for all $s'\in G_{\tau}\setminus \{s\}$), and $k = \beta/t$. We need that $P^{k_0}(s',s)\geq \varepsilon$ for each $s'\in G_{\tau}$ and a suitable $\varepsilon>0$. Since $C$ is recurrent, there is a path of at most $|C|-1<|S|$ hops in the Markov chain between any pair of states $s,s'\in C$. Because non-zero transition probabilities in $P$ are at least $p_0$, $\varepsilon:=p_0^{|S|}$ is a feasible choice.
	
	We conclude that Lemma~\ref{lem:rosenthal} yields that there is a unique stationary distribution $\pi_{\tau}$ under $P^t$ on $G_{\tau}$ satisfying that 
	\begin{equation*}
	\|\pi_{r+\beta,s} - \pi_{\tau}\| \leq (1-\epsilon)^{\lfloor k/k_0\rfloor} = \left(1-p_0^{|S|}\right)^{c \ln(D+n) / p_0^{|S|}} \leq e^{-c \ln D} = 1/D^c.
	\end{equation*}
\end{proof}

From this corollary, we infer the following coupling result.
\begin{lemma}
\label{lem:pi_s}
	Let $1 \leq i \leq \beta$ and $\tau = i \bmod t$. Then, for any state $s \in G_{\tau}$ and any constant $c>0$, there exists a probability distribution $\pi_s$ such that
\begin{equation*}
\forall r \in B_i, r \leq \Delta - \beta:~ \frac{\pi_s}{D^c} + \left(1 - \frac{1}{D^c}\right)\pi_{\tau} = \pi_{r+\beta,s},
\end{equation*}
where $\pi_{\tau}$ is the unique stationary distribution of $G_{\tau}$ under $P^t$.
\end{lemma}
\begin{proof}
If $G_{\tau}=\{s\}$, trivially $\pi_{\tau}(s)=\pi_{r+\beta,s}(s)=1$ and we choose $\pi_s(s):=1$. Thus, assume that $|G_{\tau}|>1$ in the following. For each $s'\in G_{\tau}$,
\begin{equation*}
\pi_{\tau}(s') = \sum_{s''\in G_{\tau}}P^t(s'',s')\pi_{\tau}(s'') \geq p_0^{|S|}\sum_{s''\in G_{\tau}}\pi_{\tau}(s'') = p_0^{|S|}\geq \frac{1}{D^{o(1)}},
\end{equation*}
where the first inequality exploits that any state $s'\in G_{\tau}\subseteq C$ must be reachable from any state $s''\in G_{\tau}\subseteq C$ by a sequence of at most $|C|-1<|S|$ state transitions. Also, note that $p_0^{|S|} = 1/D^{o(1)}$ because by assumption we are guaranteed that $\ell + \log b \leq \log \log D - \omega(1)$. Since $|G_{\tau}|>1$, this also implies that $\pi_{\tau}(s')\leq 1-1/D^{o(1)}$ for each $s'\in G_{\tau}$.

	Now use the equation $(1/D^c) \pi_s + (1 - 1/D^c) \pi_{\tau} = \pi_{r+\beta,s}$ to define $\pi_s$, i.e.,
	\begin{equation}
	\forall s'\in G_{\tau}:~\pi_s(s'):=  D^c\left(\pi_{r+\beta,s}(s')-\left(1 - \frac{1}{D^c}\right) \pi_{\tau}(s')\right).\label{eq:def_pi_s}
	\end{equation}
	We need to show that $\pi_s$ indeed is a probability distribution. It holds that
	\begin{equation*}
	\sum_{s'\in G_{\tau}}\pi_s(s')=D^c\left(\sum_{s'\in G_{\tau}}\pi_{r+\beta,s}(s')-\left(1 - \frac{1}{D^c}\right) \sum_{s'\in G_{\tau}}\pi_{\tau}(s')\right)=1.
	\end{equation*}
	
	Hence it remains to show that for each $s'\in G_{\tau}$, we have that $0\leq \pi_s(s')\leq 1$. For $s'\in G_{\tau}$, we bound
	\begin{eqnarray*}
	\pi_s(s') &=& D^c\left(\pi_{r+\beta,s}(s')-\left(1 - \frac{1}{D^c}\right) \pi_{\tau}(s')\right)\\
	&\leq & D^c\|\pi_{r+\beta,s}-\pi_{\tau}\|+\pi_{\tau}(s')\\
	&\leq & \frac{1}{D^{2c}}+1-\frac{1}{D^{o(1)}}\\
	&\leq & 1.
	\end{eqnarray*}
	Here we use that, by Corollary~\ref{cor:bound}, $\pi$ and $\pi_{r+\beta,s}$ are approximately equivalent (using $3c$ as the constant in the exponent) in the third step, as well as that $D$ is sufficiently large. Similarly,
\begin{equation*}
\pi_s(s') \geq \frac{\pi_{\tau}(s')}{D^c}-D^c\|\pi_{r+\beta,s}-\pi_{\tau}\|
\geq \frac{1}{D^{o(1)}D^c}-\frac{1}{D^{2c}} \geq 0.
\end{equation*}
This shows that Equation~\eqref{eq:def_pi_s} indeed ensures that $\pi_s$ is a probability distribution, concluding the proof.
\end{proof}

We now show that within each class $B_i$, approximating the random walk of an agent in the Markov chain by drawing its state for each round $r\in B_i$ \emph{independently} from the stationary distribution $\pi_{\tau}$ does not introduce a substantial error. To this end, consider the following random experiment $E_i$. For each round $r+\beta\leq \Delta$, $r\in B_i$, we toss an independent biased coin such that it shows head w.h.p. In this case, we draw the state at the end of round $r$ independently from $\pi_{\tau}$. Otherwise, we draw it from the distribution $\pi_s$, where $s$ is the state the agent was in $\beta$ steps ago and $\pi_s$ is the distribution given by Lemma~\ref{lem:pi_s}. Since each time the coin shows head w.h.p., a union bound shows that it holds that w.h.p.\ the coin shows head in \emph{all} rounds $r\in B_i$.

\begin{corollary}\label{cor:coin}
	If the experiment $E_i$ described above is executed with probability of heads being $1-1/D^{c+2}$, where $c$ is the constant from Definition~\ref{def:whp}, w.h.p.\ no coin flip shows tail. In other words, for each round $r$, the state of the agent at the end of round $r+\beta\leq \Delta$, $r\in B_i$, is drawn independently from the stationary distribution $\pi_{\tau}$ (of $G_{\tau}$ with respect to the chain induced by $P^t$).
\end{corollary}

\begin{proof}
For each round $r+\beta$ as specified by the corollary, we apply Lemma~\ref{lem:pi_s} with parameter $c+2$ and $s$ being the state of the agent at the end of round $r$. Hence, there exists a distribution $\pi_s$ such that
\begin{equation*}
\frac{\pi_{s}}{D^{c+2}} + \left(1 - \frac{1}{D^{c+2}}\right) \pi_{\tau} = \pi_{r+\beta,s},
\end{equation*}
where $s$ is the state of the agent at the end of round $r$. In other words, the following random experiment is equivalent to drawing from $\pi_{r+\beta,s}$.
\begin{compactenum}
\item Flip a biased coin showing heads with probability $1-1/D^{c+2}$.
\item If it shows heads, draw from $\pi_{\tau}$.
\item Otherwise, draw from $\pi_s$.
\end{compactenum}
Now consider all the coin flips during rounds $r+\beta\leq \Delta$, $r\in B_i$. By a union bound, the probability that no coin flip ever results in tails is bounded from below by $\Delta/D^{c+2}\leq D^2/D^{c+2}\leq 1/D^c$.
\end{proof}

\subsubsection{Movement on the grid.}
\label{sec:movement}

Having established that an agent's state can essentially be understood as a (sufficiently small) collection of sets of independent random variables, we focus on the implications on the agents' movement in the grid. Let the random variable $X_r^{\uparrow}$ have value $1$ if the state of the agent at the end of round $r$ is labeled up, and $0$ otherwise. Note that these random variables depend only on the state transitions the agent performs in the Markov chain. Also let $X_{\leq r}^{\uparrow} = \sum^{r}_{r'=1} X_{r'}^{\uparrow}$.

\begin{lemma}
\label{lem:upwards}
	Suppose agent $a$ is initially in a state from the recurrent class $C=C(a)$. Then there is a $p^{\uparrow}\in [0,1]$, such that for each round $r\leq \Delta$ it holds that w.h.p.\ $\left|X_{\leq r}^{\uparrow}-rp^{\uparrow}\right| = o(D/|S|)$.
\end{lemma}

\begin{proof}
Throughout the entire proof, we condition on the agent being initially in a state from $C$, but for simplicity largely omit this from the notation. Denote by $\E_C$ the conditional expectation given that agent $a$ is initially in recurrent class $C$. We will consider the special case $r=\Delta$ first. Since by assumption $\chi({\cal A}) \leq \log \log D - \omega(1)$, it follow that $|S| \leq 2^b \leq o(\log D)$, and so $\beta=o(D/|S|)$.
Therefore, it suffices to show that, for a suitable choice of $p^{\uparrow}$, $\left|\sum_{r'=\beta+1}^{\Delta}X_{r'}^{\uparrow}-(r-\beta)p^{\uparrow}\right|=o(D/|S|)$. Recall that $B_i$ is the collection of step numbers $i + j \beta \leq \Delta$ for $j \in \N_0$. Observe that
	\begin{equation*}
	\sum_{r'=\beta+1}^{\Delta}X_{r'}^{\uparrow}=\sum_{i=1}^\beta\sum_{\substack{\beta+1\leq r'\leq \Delta\\ r'\in B_i}}X_{r'}^{\uparrow}.
	\end{equation*}
	Denote for each $B_i$ by ${\cal E}_i$ the event that experiment $E_i$ results in all coin flips showing head. By Corollary~\ref{cor:coin}, this occurs w.h.p., and by a union bound $\bigwedge_i {\cal E}_i$ occurs w.h.p. We will show that for each $i$, conditioned on ${\cal E}_i$, w.h.p.\ it holds that
	\begin{equation}\label{eq:condition}
	\left|\sum_{\substack{\beta+1\leq r'\leq \Delta\\ r'\in B_i}}X_{r'}^{\uparrow}
-\E_C\left[\sum_{\substack{\beta+1\leq r'\leq \Delta\\ r'\in B_i}}X_{r'}^{\uparrow}\,\Bigg|\,{\cal E}_i\right]\right|= o\left(\frac{D}{|S|\beta}\right)
	\end{equation}
	Conditioned on ${\cal E}_i$, we know that the considered variables $X_{r'}^{\uparrow}$ from $B_i$ are independently and identically distributed: The state at the end of round $r'$ is drawn independently from some stationary distribution $\pi_{\tau}$ that does not depend on $r'$, and the probability for the agent to move up in the grid equals the probability that this state is labeled ``up". Denote by $p_i^{\uparrow}$ the probability for the agent to move up in such a round $r'$ (when its state is distributed according to $\pi_{\tau}$). By linearity of expectation,
	\begin{equation*}
	\mu_i:=\E_C\left[\sum_{\substack{\beta+1\leq r'\leq \Delta\\ r'\in B_i}}X_{r'}^{\uparrow}\,\Bigg|\,{\cal E}_i\right]
	=\sum_{\substack{\beta+1\leq r'\leq \Delta\\ r'\in B_i}}\E_C\left[X_{r'}^{\uparrow}\,\big|\,{\cal E}_i\right]
	=\sum_{\substack{\beta+1\leq r'\leq \Delta\\ r'\in B_i}}p_i^{\uparrow}
	=p_i^{\uparrow} \left(\frac{\Delta}{\beta}-1\right).
	\end{equation*}

	If $\mu_i\leq 3c\ln D=o(D/(\beta|S|))$, Chernoff's bound (Inequality~\eqref{eq:chernoff_upper} with $\delta=1$) implies that
	\begin{equation*}
	P\left[\sum_{\substack{\beta+1\leq r'\leq \Delta\\ r'\in B_i}}X_{r'}^{\uparrow}\geq 6c\ln D\,\Bigg|\,{\cal E}_i\right]
	\leq e^{-c\ln D}\leq \frac{1}{D^c}.
	\end{equation*}
	On the other hand, if $\mu_i> 3c\ln D$, we choose $\delta:=\sqrt{3c\ln D/\mu_i}<1$ and apply the two-sided Chernoff bound (Inequality~\eqref{eq:chernoff_twosided}), yielding that
	\begin{equation*}
	P\left[\left|\sum_{\substack{\beta+1\leq r'\leq \Delta\\ r'\in B_i}}X_{r'}^{\uparrow}
	-\mu_i\right|>\delta \mu_i\,\Bigg|\,{\cal E}_i\right]
	\leq 2e^{-\delta^2\mu_i/3}= \frac{2}{D^c}.
	\end{equation*}
	We fix\footnote{As stated earlier, we deferred the choice of the $o(1)$-term in the exponent of $\Delta$, permitting to specify $\Delta$ in terms of $\beta$ now.} $\Delta:=o(D^2/(\beta |S|^2 \log D))=D^{2-o(1)}$. Thus, it holds that
	\begin{equation*}
	\delta \mu_i = \sqrt{3c\ln D\mu_i} = \BO\left(\sqrt{\frac{p_i^{\uparrow} \Delta \log D}{\beta}}\right)
	= o\left(\frac{D}{|S|\beta}\right).
	\end{equation*}
	This shows that, conditioned on ${\cal E}_i$, the bound \eqref{eq:condition} holds w.h.p. To complete our line of reasoning, observe that
	\begin{equation*}
	\E_C\left[X_{\leq \Delta}^{\uparrow}\,\Big |\,\bigwedge_{i=1}^{\beta} {\cal E}_i\right]=\sum_{i=1}^{\beta}\mu_i \pm \BO(\beta)= \Delta\sum_{i=1}^{\beta}\frac{p_i^{\uparrow}}{\beta}\pm \BO(\beta)=\Delta\sum_{i=1}^{\beta}\frac{p_i^{\uparrow}}{\beta}\pm o\left(\frac{D}{|S|}\right)
	\end{equation*}
	and set $p^{\uparrow}:=\sum_{i=1}^{\beta}p_i^{\uparrow}/\beta$.	By a union bound, we have that, w.h.p., both $\bigwedge_i {\cal E}_i$ occurs and bound~\eqref{eq:condition} holds for all $i$. In this case, it follows that
	\begin{eqnarray*}
	\left|X_{\leq r}^{\uparrow}-rp^{\uparrow}\right|&=&\left|X_{\leq r}^{\uparrow}-\E_C\left[X_{\leq \Delta}^{\uparrow}\,\Big |\,\bigwedge_{i=1}^{\beta} {\cal E}_i\right]\right|- \left|\E_C\left[X_{\leq \Delta}^{\uparrow}\,\Big |\,\bigwedge_{i=1}^{\beta} {\cal E}_i\right]-rp^{\uparrow}\right|\\
	&\leq & \sum_{i=1}^{\beta}o\left(\frac{D}{|S|\beta}\right)+\BO(\beta) + o\left(\frac{D}{|S|}\right)\\
	&=& o\left(\frac{D}{|S|}\right).
	\end{eqnarray*}
	This proves the claim for the special case $r=\Delta$. For the general case, observe that decreasing $r$ by an integer multiple of $\beta$ will decrease the computed expectation by the same multiple of $p^{\uparrow}$ (as long as $r$ remains larger than $\beta+1= o(D/|S|)$). Concerning the bound~\eqref{eq:condition}, observe that decreasing the number of steps will only decrease the probability of large deviations from the expectation of the random variable.	Since $\beta=o(D/|S|)$, the general statement hence follows analogously to the special case.
\end{proof}

Repeating these arguments for the other directions (right, down, and left), we see that overall, each agent behaves fairly predictably. Define $X_{\leq r}\in \Z^2$ to be the random variable describing the sum of all moves the agent performs in the grid up to round $r$, i.e., its position in the grid (in each dimension) at the end of round $r$. For this random variable, the following statement holds.
\begin{corollary}\label{cor:concentration}
	Suppose agent $a$ is initially in a state of the recurrent class $C(a)$. Then there is $\vec{p}\in [0,1]^2$ depending only on $C(a)$ such that for each $r\leq \Delta$, it holds that $\left\|X_{\leq r}-r\vec{p}\,\right\|= o(D/|S|)$ w.h.p.
\end{corollary}

\begin{proof}
Defining the random variables $X^{\rightarrow}_{\leq r}$, $X^{\downarrow}_{\leq r}$, and $X^{\leftarrow}_{\leq r}$ analogously to $X^{\uparrow}_{\leq r}$, we can apply the same reasoning as in Lemma~\ref{lem:upwards} to control their values. Hence, for each $X^{\bullet}_{\leq r}$, $\bullet\in \{\uparrow,\rightarrow,\downarrow,\leftarrow\}$, there is a $p^{\bullet}\in [0,1]$ (depending only on $C(a)$) such that $\left|X_{\leq r}^{\bullet}-rp^{\bullet}\right| = o(D/|S|)$ w.h.p. Observe that $X_{\leq r}=(X_{\leq r}^{\uparrow}-X_{\leq r}^{\downarrow},X_{\leq r}^{\rightarrow}-X_{\leq r}^{\leftarrow})$. Hence, with $\vec{p}:=(p^{\uparrow}-p^{\downarrow},p^{\rightarrow}-p^{\leftarrow})$, a union bound shows that $\left\|X_{\leq r}-r\vec{p}\,\right\| = o(D/|S|)$ w.h.p.
\end{proof}

We are now ready to resume the proof of Theorem \ref{thm:lower}.
\begin{proof}[Proof of Theorem~\ref{thm:lower}]
Denote by ${\cal C}$ the set of recurrent classes of the Markov chain describing an agent's state evolution. By Corollary~\ref{cor:initial}, it holds for each agent $a$ that, after each round $r \geq R_0=D^{o(1)}$, w.h.p.\ the agent is located in recurrent class $C(a)\in {\cal C}$. Since Lemma~\ref{lem:upwards}, and therefore Corollary~\ref{cor:concentration}, do not depend on the initial state from $C(a)$ the agent is in, the same reasoning shows that, at the end of round $r$, w.h.p.\ the position of $a$ will not deviate by more than distance $o(D/|S|)$ from a straight line in the grid. By a union bound, this holds for all agents jointly w.h.p (recall that by assumption $n$ is sub-exponential in $D$). Hence, w.h.p., it holds for each agent $a$ and each round $r \geq R_0$ that $a$ never ventures further away from the origin than distance $o(D/|S|)$, or its position does not deviate by more than distance $o(D/|S|)$ from one of at most $|{\cal C}|$ straight lines or the origin. Since for any straight line only a segment of length $\BO(D)$ is in distance $\BO(D)$ from the origin, the union of all grid points that are (i) in distance at most $D$ from the origin and (ii) in distance at most $o(D/|S|)$ from one of the straight lines has cardinality $\BO(D)\cdot o(D/|S|)\cdot |{\cal C}|\leq o(D^2/|S|)\cdot |S|= o(D^2)$. Hence, there is a set $G\subset \Z^2$ of $o(D^2)$ grid points that only depends on the algorithm ${\cal A}$, $D$, and $n$, with the following property: w.h.p., all grid points in distance $D$ from the origin that are visited within the first $\Delta$ steps of an execution of ${\cal A}$ are in $G$. Since there are $\Theta(D^2)$ grid points in distance $D$ from the origin, this implies that the target can be placed in such a way that w.h.p.\ no agent will find it within $\Delta=D^{2-o(1)}$ rounds, and a uniformly placed target is found in this amount of time with probability $o(1)$.
\end{proof}

Finally, we need to show that Theorem \ref{thm:lower} also holds with respect to the metric $M_{\text{moves}}$. In the following corollary, we show that each move of an agent on the grid corresponds to at most $D^{o(1)}$ transitions in its Markov chain, or otherwise, the agent does not move on the grid after some point on. 

\begin{corollary}
\label{cor:moves}
	Let $\mathcal{A}$ be an algorithm with $\chi({\cal A}) = b + \log \ell \leq \log \log D-\omega(1)$ and $n \in poly(D)$ agents. There is a placement of the target within distance $D$ from the origin such that w.h.p.\ no agent executing algorithm $\mathcal{A}$ finds the target in fewer than $D^{2-o(1)}$ moves.
	Moreover, the probability for some agent to find a target, placed uniformly at random in the square of side $2D$ centered at the origin, within $D^{2-o(1)}$ moves is $o(1)$.
\end{corollary}
\begin{proof}
	First, we show that w.h.p., at least one of the following is true about any agent $a$ in any round $r \geq R_0$: (1) each move on the grid performed by $a$ corresponds to at most $D^{o(1)}$ steps in its Markov chain, or (2) $a$ is located in a recurrent class in which all states are labeled ``none".
	
	By Corollary \ref{cor:initial}, we know that at the end of round $r \geq R_0$ w.h.p.\ agent $a$ is in some recurrent class $C(a)$ that it never subsequently leaves. By Lemma \ref{cor:visits}, we know that for each state $s \in C(a)$, w.h.p.\ it is visited within $R_0 = D^{o(1)}$ steps. Therefore, if $C(a)$ contains a state labeled up/down/left/right, w.h.p.\ it is visited within $R_0 = D^{o(1)}$ steps. Otherwise, (2) applies. 
	
	If part (2) of the statement applies, then an agent does not make any progress in the grid after it reaches its recurrent class, so it does not visit more than $R_0 = D^{o(1)}$ grid points and, consequently, the corollary holds. If part (1) applies, since each move in the grid corresponds to at most $D^{o(1)}$ steps in the Markov chain, $D^{2-o(1)}$ moves correspond to $D^{2-o(1)}$ steps. In this case, we know Theorem \ref{thm:upper} guarantees that there is a placement of the target within distance $D$ from the origin such that w.h.p.\ no agent finds it within $D^{o(1)}$ steps, and consequently $D^{o(1)}$ moves.
\end{proof}

\section{Discussion and Conclusion}
\label{sec:discussion}

We have presented an algorithm and a lower bound for the problem of $n$ agents searching in a grid for a target placed at distance at most $D$ from the origin. Our lower bound shows that for $n$ sub-exponential in $D$, the agent cannot find the target w.h.p.\ in fewer than $D^{2-o(1)}$ rounds if $\chi({\cal A}) < \log \log D - \omega(1)$. We also present an algorithm that finds the target in $(D+D^2/n) 2^{\BO(\ell)}$ rounds in expectation for $\chi({\cal A}) \leq 3 \log \log D$, proving our lower bound to be near tight.

Note that for the upper bound we get stronger results if we consider a fixed search area of distance $D$ from the origin, as opposed to keeping track of a varying estimate of the search area. In the case of a fixed distance $D$, and $\chi({\cal A}) = \log \log D + \BO(1)$, suffices to find the target in $(D+D^2/n) 2^{O(\ell)}$ rounds in expectation.

Finally, we should mention that currently there is a gap of $2^{O(\ell)}$ between our upper and lower bounds. Also, for our lower bound we assume that $\chi({\cal A}) = b + \log \ell < \log \log D - \omega(1)$ while our algorithm we assume $b = 3(\log \log D -\log \ell) + O(1)$ bits of memory, which constitutes a constant-factor gap. Closing either of these gaps is an open problem.

\bibliographystyle{plain}
\bibliography{ants}

\appendix

\section{Math Preliminaries}
\label{sec:math}

In this section, we briefly go over some mathematical definitions and results that will be used throughout the proofs. Throughout the paper, by $\|\cdot\|$ we will denote the $\infty$-norm on the respective space.

\subsection{Markov Chains}

First, we state a basic result from \cite{feller68} about periodic Markov chains.

\begin{theorem}[Feller]
\label{thm:feller}
	In an irreducible Markov chain with period $t$ the states can be divided into $t$ mutually exclusive classes $G_0, \cdots, G_{t-1}$ such that it is true that (1) if $s \in G_{\tau}$ then the probability of being in state $s$ in some round $r \geq 1$ is $0$ unless $r = \tau + vt$ for some $v \in \mathbb{N}$, and (2) a one-step transition always leads to a state in the right neighboring class (in particular from $G_{t-1}$ to $G_0$). In the chain with matrix $P^t$ each class $G_{\tau}$ corresponds to an irreducible closed set.
\end{theorem}

The next theorem establishes a bound on the difference between the stationary distribution of a Markov chain and the distribution resulting after $k$ steps.

\begin{lemma}[Rosenthal]
\label{lem:rosenthal}
	Let $P(x,\cdot)$ be the transition probabilities for a time-homogeneous Markov chain on a general state space $\mathcal{X}$. Suppose that for some probability distribution $Q(\cdot)$ on $\mathcal{X}$, some positive integers $k$ and $k_0$, and some $\epsilon > 0$, 
	\begin{equation*}
		\forall x \in \mathcal{X}:\,P^{k_0}(x,\cdot) \geq \epsilon Q(\cdot) ,
	\end{equation*}		
where $P^{k_0}$ represents the $k_0$-step transition probabilities. 
Then for any initial distribution $\pi_0$, the distribution $\pi_k$ of the Markov chain after $k$ steps satisfies

	\begin{equation*}
		\|\pi_k - \pi\|\leq (1 - \epsilon)^{\lfloor k/k_0 \rfloor}
	\end{equation*} 
where $\| \cdot \|$ is total variation distance and $\pi$ is any stationary distribution. (In particular, the stationary distribution is unique.)
\end{lemma}

\subsection{Basic Probability Definitions and Results}


\begin{definition}
\label{def:whp}
	Let $\pi$ be some probability distribution and let $E$ be an arbitrary event in $\pi$. For a given $N\in \N$, we say that event $E$ occurs ``with high probability in $N$'' iff the probability of the event $E$ occurring is at least $1 - 1 / N^c$ for an arbitrary predefined constant $c > 0$.
\end{definition}

Next, we provide a similar definition for the distance between two probability distributions.

\begin{definition}
	Let $\pi_1$ and $\pi_2$ be two probability distributions with the same range. For a given $N \in \N$, we say that $\pi_1$ and $\pi_2$ are ``approximately equivalent with respect to $N$'' iff $\|\pi_1 - \pi_2\| = \BO(1/N^c)$ for an arbitrary predefined constant $c > 0$.
\end{definition}	
	Finally, we state the two versions of Chernoff's bound that we use throughout the proofs.
	
\begin{theorem}[Chernoff bound] Let $X_1, \cdots, X_k$ be independent random variables such that for $1 \leq i \leq k$, $X_i \in\{0,1\}$. Let $X = X_1 + X_2 + \cdots + X_k$ and let $\mu = \E[X]$. Then, for any $0 \leq \delta \leq 1$, it is true that:

\begin{equation}
	P[X > (1 + \delta) \mu] \leq e^{-\delta^2 \mu/2}\label{eq:chernoff_upper}
\end{equation}	
	
\begin{equation}
	P[X < (1 - \delta) \mu] \leq e^{-\delta^2 \mu/3}\label{eq:chernoff_lower}
\end{equation}	
\end{theorem}

\begin{theorem}[Two-sided Chernoff bound] Let $X_1, \cdots, X_k$ be independent random variables such that for $1 \leq i \leq k$, $X_i \in\{0,1\}$. Let $X = X_1 + X_2 + \cdots + X_k$ and let $\mu = \E[X]$. Then, for any $0 \leq \delta \leq 1$, it is true that:

\begin{equation}
	P[|X - \mu| < \delta \mu] \leq 2e^{-\delta^2 \mu/3}\label{eq:chernoff_twosided}
\end{equation}	
	
\end{theorem}

\end{document}